\documentclass[12pt]{article}
\usepackage{amsmath,amsfonts,amsthm,amssymb}
\usepackage{url}
\usepackage{graphicx}
\usepackage[font=small]{caption}
\usepackage{subcaption}
\usepackage[colorlinks=true]{hyperref}
\hypersetup{
    linkcolor = {blue},
    citecolor = {blue}
}
\usepackage{cleveref}
\usepackage{enumitem}
\usepackage{algpseudocode}
\usepackage{algorithmicx}
\usepackage{physics}
\usepackage{bbm}
\usepackage{array}
\usepackage{xfrac}
\usepackage{tikz}
\usepackage{float}
\usepackage{enumitem}
\usepackage{color}
\usepackage{mathtools}
\usepackage{comment}
\usepackage[numbers]{natbib}
\usepackage[letterpaper, top=1in, bottom=1in, right=1in, left=1in]{geometry}
\usepackage{thm-restate}
\usepackage{makecell}

\usepackage{csquotes}
\usepackage[dvipsnames]{xcolor}
\usepackage[suppress]{color-edits}  % use this to suppress the package
\addauthor[Bob]{ret}{blue}    % ret for Bob
\addauthor[Xiaoyang]{xxy}{brown!80!black}
\addauthor[Lazy]{lazy}{green!60!black}

\graphicspath{{images/}}

\newcommand{\be}[0]{\begin{enumerate}}
\newcommand{\ee}[0]{\end{enumerate}}
\newcommand{\bi}[0]{\begin{itemize}}
\newcommand{\ei}[0]{\end{itemize}}

\newcommand{\OO}{\mathrm{O}}

\newcommand{\Decreasekey}{\mbox{\it decrease-key}}
\newcommand{\Delete}{\mbox{\it delete}}

\newcommand{\Deletemin}{\mbox{\it delete-min}}

\newcommand{\Findmin}{\mbox{\it find-min}}

\newcommand{\Insert}{\mbox{\it insert}}

\newcommand{\Makeheap}{\mbox{\it make-heap}}
\newcommand{\Meld}{\mbox{\it meld}}

\newtheorem{theorem}{Theorem}[section]

\newtheorem{lemma}[theorem]{Lemma}
\newtheorem{corollary}[theorem]{Corollary}

\begin{document}
\title{Pure Pairing Heaps}
\author{Robert E. Tarjan\thanks{Department of Computer Science, Princeton University.  Research partially supported by a gift from Microsoft.} \and Xiaoyang Xu\thanks{Department of Computer Science, University of California, Santa Barbara}}
%\author{Anonymous Author(s)}

\date{}

\maketitle

\begin{abstract}
The \emph{pairing heap}~\cite{FSST86} is a ``self-adjusting'' implementation of a heap (priority queue) that is widely used in practice because it is simple and efficient.  We introduce and analyze a simplified version of the pairing heap that we call the \emph{pure pairing heap}.  Our innovation is to eliminate the assembly pass during delete-min operations.  We obtain the following amortized time bounds for operations on pure pairing heaps: $\OO(\log n)$ time per delete-min, $\OO(\log\log n \cdot \log\log\log n)$ time per decrease-key operation, and $\OO(1)$ time for each insert or meld.  These bounds match those recently obtained for a more complicated version of pairing heaps, the \emph{multipass pairing heap}.  These bounds also match the known lower bounds for self-adjusting heaps, except for the decrease-key bound, which is within a factor of $\log\log\log n$ of the lower bound.  The main novelty in our analysis is to partition heap items into groups and to analyze each group separately.  Our analysis extends to give the same bounds for lazy pairing heaps, a multitree version of pairing heaps.    
\end{abstract}

\section{Introduction}\label{S:introduction}

A \emph{heap} (or \emph{priority queue}) is a data structure containing a set of items, each with an associated key selected from a totally ordered key space, that supports the following operations:

\begin{itemize}
\item[] $\Makeheap()$: Create and return a new, empty heap.

\item[] $\Findmin(H)$: Return an item of smallest key in heap $H$; return null if $H$ is empty.

\item[] $\Insert(H, v)$: Insert item $v$ with predefined key into heap $H$.  Item $v$ must be in no other heap.

\item[] $\Deletemin(H)$: Delete $\Findmin(H)$ from $H$.

\item[] $\Meld(H_1, H_2)$: Return a heap containing all items in item-disjoint heaps $H_1$ and $H_2$, destroying $H_1$ and $H_2$ in the process.

\item[] $\Decreasekey(H, v, k)$: Given the location of item $v$ in heap $H$, and given that the key of $v$ is at least $k$, decrease the key of $v$ in $H$ to $k$. 

\item[] $\Delete(H, v)$: Given the location of item $v$ in heap $H$, delete $v$ from $H$. 
\end{itemize}

One can implement $\Delete(H, v)$ as $\Decreasekey(H, v, -\infty)$ followed by $\Deletemin(H)$, where $-\infty$ is a value smaller than those of all actual keys.  We shall assume this implementation and not further discuss $\Delete$.  

If the only operations on keys are comparisons, $n$ insertions followed by $n$ delete-min operations will sort $n$ items by key, so the amortized time\footnote{We shall study \emph{amortized time} throughout. See~\cite{tarjan1985amortized} for a discussion of amortized analysis.} of either insertion or delete-min or both must be $\Omega(\log n)$ on an $n$-item heap.  The \emph{Fibonacci heap}~\cite{Fibonacci}  supports delete-min in $\OO(\log n)$ amortized time and all the other heap operations in $\OO(1)$ amortized time.  Since the invention of the Fibonacci heap, many other heaps with similar efficiency have been developed. See~\cite{BrodalSurvey, HollowHeaps, KS19}.

All these data structures obtain their efficiency by maintaining a balance condition.  An alternative is to design a ``self-adjusting'' data structure, one that does not maintain an explicit balance condition but instead obtains its efficiency by improving the structure enough during slow operations to make slow operations rare.  One such self-adjusting heap is the \emph{pairing heap}~\cite{FSST86}.  Pairing heaps are simple and efficient in practice~\cite{LarkinSenTarjan, StaskoVitter}.  The paper that introduced pairing heaps proved an $\OO(\log n)$ amortized time bound for all operations.  The paper conjectured that pairing heaps have the same amortized efficiency as Fibonacci heaps ($\OO(1)$ time for each heap operation other than delete-min), but this was disproved by Fredman~\cite{FredmanLB}, who showed that pairing heaps and similar self-adjusting heaps must take $\Omega(\log\log n)$ time per decrease-key if they take $\OO(\log n)$ time per delete-min.  Iacono and {\"O}zkan~\cite{IaconoOzkan} proved the same lower bound for a different large class of self-adjusting heaps. These results raise the question of whether the pairing heap, or any other kind of self-adjusting heap, has efficiency matching the lower bounds.  Sinnamon and Tarjan~\cite{SinnamonT25} answered this question in the affirmative for two much-more-recently-discovered, closely-related heap implementations, slim and smooth heaps.  They also showed that the multipass pairing heap, a variant of the pairing heap, has efficiency matching the lower bounds except for the decrease-key operation, for which they obtained a bound of $\OO(\log\log n \cdot \log\log\log n)$, off by a factor of $\log\log\log n$ from the lower bound.

We introduce a simplified version of the pairing heap that we call the \emph{pure pairing heap}.  Our idea is to do a delete-min by a single pairing pass, thereby eliminating the assembly pass done in standard pairing heaps.  The pure pairing heap is a one-tree version of a multitree variant described in the original paper on pairing heaps~\cite{FSST86}, there called a \emph{lazy pairing heap}.  A more recent paper~\cite{purepairingheap2021} used the term ``pure pairing heap'' for this lazy version, but we prefer to reserve this term for the one-tree version, which is simpler and more in the spirit of the standard pairing heap.

We analyze pure pairing heaps by extending Sinnamon and Tarjan's analysis of multipass pairing heaps to them.  We obtain amortized time bounds of $\OO(1)$ per insert and meld, $\OO(\log n)$ per delete-min, and $\OO(\log\log n \cdot \log\log\log n)$ per decrease-key.  Here $n$ is the number of items that are in the heap or heaps just before the operation and that are eventually deleted.  (Items that are never deleted do not count in this bound.)  These bounds are tight except for the decrease-key bound, which is a factor of $\log\log\log n$ from being tight.

The main novelty in our analysis is to group certain nodes together and study the links between them.  If we treat each group in isolation, the multipass analysis applies to the links within a group.  We also need to bound the number of groups and the number of links between nodes in different groups.  This requires additional new ideas.

Our analysis of pure pairing heaps also applies to lazy pairing heaps and gives essentially the same bounds, depending on the particular version of lazy pairing heap.

The remainder of our paper consists of nine sections.  Section~\ref{S:versions-of-pairing-heaps} describes the versions of pairing heaps that we study in detail.  Section~\ref{S:terminology} gives some basic terminology about these data structures.  Section~\ref{S:node-groups} presents our key new idea, that of \emph{node groups}.  Section~\ref{S:pairing-links} develops properties of pairing links, which are the primitive operations we need to count.  Section~\ref{S:node-ranks} studies \emph{node ranks}, which are our main tool for obtaining amortized bounds on the number of pairing links.  Section~\ref{S:link-ranks} studies \emph{link ranks}, which help measure the effect of pairing links on node ranks.  Section~\ref{S:winner-links} bounds the number of \emph{winner links}, the critical subset of links we need to count, and completes our analysis of pure pairing heaps.  Section~\ref{S:lazy-pairing-heaps} extends our results to lazy pairing heaps.  We conclude in Section~\ref{S:remarks} with some final remarks.    

\section{Versions of pairing heaps}\label{S:versions-of-pairing-heaps}

We consider heap implementations that represent a heap by one or more rooted trees whose nodes are the heap items. Such a representation of a heap is an \emph{endogenous} data structure: The heap items \emph{are} the tree nodes, rather than being stored \emph{in} the tree nodes. Henceforth we shall speak of nodes rather than items.  Each tree is heap-ordered by node key: If node $v$ with key $v.key$ is the parent of node $w$ with key $w.key$, then $v.key \leq w.key$.  Thus the root of a tree is a node of minimum key in the tree.  Access to a tree is via its root.

The fundamental primitives for modifying rooted, heap-ordered trees are \emph{linking} and its inverse, \emph{cutting}.  A \emph{link} of the roots of two node-disjoint heap-ordered trees combines the trees by making the root of smaller key the parent of the other root, breaking a tie arbitrarily.  The surviving root is the \emph{winner} of the link; the new child is the \emph{loser} of the link.  We denote by $vw$ a link won by $v$ and lost by $w$.

The children of each node form an ordered list.  The order depends on the way links are done, as we shall describe.  We think of this order as being left to right: the first and last child of a node are its \emph{leftmost} and \emph{rightmost} child, respectively.

A \emph{cut} of a link $vw$ breaks the link, breaking the tree containing $v$ and $w$ into two trees, one rooted at $w$ and containing the subtree of $w$ as it existed before the cut, and one containing $v$ that is unchanged except for the removal of $w$ and its descendants.

In the heap implementations we consider, each link makes the loser the new \emph{leftmost} child of the winner.  This orders the children of a node by link time, latest leftmost.

The original version of the pairing heap, which we call \emph{the standard pairing heap}, or just the \emph{pairing heap}, represents a heap by a single rooted, heap-ordered tree.  The heap operations are done as follows:

\begin{itemize}
\item[] $\Makeheap()$: Create and return an empty tree.

\item[] $\Findmin(H)$: Return the root of $H$.

\item[] $\Insert(H, v)$: Make $v$ into a one-node heap.  Meld this heap with $H$.

\item[] $\Meld(H_1, H_2)$: If either $H_1$ or $H_2$ is empty, return the other; otherwise, return the tree formed by linking the roots of $H_1$ and $H_2$. 

\item[] $\Decreasekey(H, v, k)$: Set $v.key = k$.  If $v$ is not the root of $H$, cut the link between $v$ and its parent, and link $v$ with the root of $H$.

\item[] $\Deletemin(H)$: Let $u$ be the root of $H$.  If $H$ has no other nodes, replace $H$ by an empty tree.  Otherwise, cut all the links between $u$ and its children.  This makes the list of children of $u$ into a list of roots.  Link these roots in two passes.  In the first pass, the \emph{pairing pass}, proceed left to right through the root list, linking the roots in adjacent pairs, the first with the second, the third with the fourth, and so on. In the second pass, the \emph{assembly pass}, repeatedly link the current rightmost root with the root to its left until only one root remains.  During both passes, each winner of a link retains its position in the list of roots; each loser is deleted from this list.  Once only one root remains, replace $H$ by the tree rooted at this root.
\end{itemize}

The \emph{multipass pairing heap} also represents a heap by a single rooted, heap-ordered tree.  It does all the heap operations except delete-min in the same way as standard pairing heaps.  To do a delete-min, a multipass pairing heap does not do an assembly pass.  Instead it does repeated pairing passes until there is only one root.

The \emph{pure pairing heap}, which we introduce here, also does all the heap operations except delete-min in the same way as standard pairing heaps. It does a delete-min via a single pairing pass but no assembly pass.  Instead, during the pairing pass it keeps track of a remaining root of minimum key, say $v$.  After the pairing pass, it makes $v$ the parent of all other remaining roots, by deleting $v$ from the list of remaining roots and then catenating this list with the list of children of $v$, with the former to the left of the latter.  We call this the \emph{assembly phase} of the delete-min.  The effect of the assembly phase is to link $v$ with each of the remaining roots other than $v$, proceeding right to left through the list of remaining roots other than $v$.  
One simple way to implement a pure pairing heap is to store three pointers with each node: to its leftmost child; to its right sibling; and to its left sibling, or to its parent if it is the leftmost sibling.  That is, each node has a pointer to a doubly linked list of its children.  With this representation, each link or cut takes $\OO(1)$ time.  The number of pointers per node can be reduced to two~\cite{FSST86}. 

A delete-min in a pure pairing heap takes just one pass through the root list:  At the end of the pairing pass, the algorithm has direct access to both a root $v$ of minimum key and to the leftmost and rightmost roots on the root list.  It deletes $v$ from the root list and catenates the list of remaining roots with the list of children of $v$.  This takes $\OO(1)$ time.

The time bounds we shall derive for pure pairing heaps hold even if we treat the links during the assembly phase as being done explicitly, rather than implicitly via a single list catenation.  Specifically, the analysis applies if we view a delete-min as doing the following: First it does a pairing pass, then it moves a root of minimum key, say $v$, to the right end of the list of roots, and finally it links each root other than $v$ to $v$, proceeding right to left through the root list.  Equivalently, it moves $v$ to the right end of the list of roots and does an assembly pass.

Our analysis also holds for the variant of pure pairing heaps in which the list of children of $v$ is catenated with the list of remaining roots, with the former to the left of the latter.  It also holds if the list of children of $v$ replaces $v$ in the list of remaining roots, and this entire list becomes the new list of children of $v$.  What matters is that the old children of $v$ stay contiguous and in the same order, and the remaining roots stay in the same order.

We call our data structure the pure pairing heap because it captures the idea of the pairing heap in its purest form.  The purpose of the assembly pass in a pairing heap is to link all the winners of the pairing links, so that the result is a single tree rather than a set of trees.  The simplest way to do this is to keep track of a remaining root of minimum key while doing the pairing pass.  Not only is the pure pairing heap arguably simpler than the standard pairing heap, but the bounds we shall prove for it match the best bounds known for any version of pairing heaps.

We also consider a multitree version of the pure pairing heap that we call the \emph{lazy pairing heap}.  A lazy pairing heap represents a heap by a list of roots of heap-ordered trees, with the first root on the list (the \emph{min-root}) having minimum key.  A find-min returns the min-root.  An insertion creates a new one-node heap and melds it with the existing heap.  A meld catenates the lists of roots of the two heaps, putting first the list of the heap whose min-root has smaller key.  A decrease-key of $v$ decrease the key of $v$.  If $v$ is a child, it cuts $v$ from its parent; if $v$ is a root other than the min-root, it deletes $v$ from the root list.  Finally, if $v$ is not the min-root, it links $v$ with the min-root. A delete-min deletes the min-root, making its children into a list of roots.  It catenates this list of roots with the list of remaining roots, with the former first.  Finally, it does a pairing pass, during which it keeps track of a remaining root of minimum key.  After the pairing pass, it makes the root of minimum key the min-root by moving it to the front of the list of roots.

This is the version of lazy pairing heaps described in~\cite{SODASmoothSlim} (and there called a pure pairing heap).  An alternative that avoids doing a link during a decrease-key is to insert the node whose key decreases into the root list, in first or second position depending on whether it becomes the min-root.  An independent alternative is to do an insertion by linking the new node with the min-root, instead of inserting it into the root list.  The version of lazy pairing heaps presented in the original pairing heap paper~\cite{FSST86} differs from these variants in that it does not maintain a min-root.  Instead, it does a find-min by doing a pairing pass and keeping track of a root of smallest key.  It does a delete-min by doing a find-min, deleting the root $u$ returned by the find-min, and catenating the list of new roots, previously the list of children of $u$, with the list of remaining roots, with the latter first.  When doing a decrease-key, it adds the node whose key decreases to the back of the list of roots rather than to the front.  We shall analyze the first version we have presented. Our analysis applies with appropriate changes to these other versions as well, although if a min-root is not maintained, the $\OO(1)$ bound for find-min becomes amortized rather than worst-case. 

\section{Node and link types}\label{S:terminology}

In our analysis, we consider a sequence of heap operations on an initially empty collection of heaps.  Each operation except delete-min does at most one link and takes $\OO(1)$ time worst-case.  A delete-min takes $\OO(1)$ time plus $\OO(1)$ time per pairing link.  It follows that the total time of a sequence of heap operations is $\OO(1)$ per operation plus $\OO(1)$ per pairing link.

Our implementation of insertion treats it as a special case of meld.  We shall assume that each meld not done by an insertion is of two non-empty heaps; that is, we ignore melds that are not done by insertions and that do not do links.  With this assumption, the total number of meld links is at most the number of insertions minus one, since there is at most one new heap per insertion, and each meld link reduces the number of heaps by one.   

We use the following terminology, adapted from~\cite{SinnamonT25} with one change:

\begin{itemize}
\setlength\itemsep{1em}

\item [] A node is \emph{temporary} if it is eventually deleted, \emph{permanent} otherwise.

\item [] A link is a \emph{meld link} (called in~\cite{SinnamonT25} an \emph{insertion link}) if it is done during a meld, a \emph{decrease-key} link if it is done during a decrease-key, a \emph{deletion link} if it is done during a delete-min.

\item [] A deletion link is a \emph{pairing link} if it is done during the pairing pass of a delete-min, an \emph{assembly link} if done during the assembly phase.      

\item [] A deletion link $vw$ is a \emph{left link} if $w$ is to the left of $v$ on the root list before the link, a \emph{right link} otherwise.  In a pure pairing heap, all the assembly links are left links.  Meld and decrease-key links are neither left nor right.

\item [] A link is a \emph{d-link} if it is cut by a delete-min, a \emph{k-link} if it is cut by a decrease-key, an \emph{f-link} if it is never cut. (The ``f'' stands for ``final''.)

\item [] A link is \emph{real} if its winner and loser are both temporary and it is a d-link, \emph{phantom} otherwise.

\item [] A child node is a \emph{real child} if it is connected to its parent by a real link, a \emph{phantom child} otherwise.  Each permanent node that is a child is a phantom child.  Each time a temporary node loses a link, it becomes a real or phantom child, depending on whether the link it loses is a d-link or a k-link: A temporary node cannot lose an f-link.
\end{itemize}

The distinction between temporary and permanent nodes allows us to bound the time per delete-min and decrease-key by functions of the number of temporary nodes instead of the total number of nodes: Insertions of permanent nodes and decrease-key operations on permanent nodes take $\OO(1)$ time, amortized as well as worst-case.

Given a heap operation, we denote by $n$ the number of temporary nodes in the heap or heaps undergoing the operation, just before the operation occurs. Since we treat an insertion as a meld, $n$ for an insertion is the number of temporary nodes in the heap after the insertion.   We denote by $\log$ the base-two logarithm.  We define $\lg n = \log(\max\{n, 1\})$.  If $f$ is a function of $n$, we use the terminology ``$f(n)$ per decrease-key'' to mean the sum of $f(n_i)$ over all decrease-key operations in a sequence of heap operations, where $n_i$ is the number of temporary nodes in the heap undergoing the $i$-th decrease-key, and similarly for ``$f(n)$ per delete-min'' and ``$f(n)$ per insertion''.  When stating order-of-magnitude bounds per operation, such as $\OO(\lg n)$, we assume an implicit ``$+1$'' inside the big $O$, to avoid writing it every time: Each operation takes $\Omega(1)$ time. 

The following lemma justifies our claim that the bounds we shall derive for pure pairing heaps hold even if the assembly links are done explicitly:

\begin{lemma}\label{L:assembly-links}
A delete-min in a pure pairing heap does at least as many pairing links as assembly links.
\end{lemma}
\begin{proof}
Consider a delete-min that does $k$ pairing links.  The number of roots remaining after the pairing pass is at most $k+1$, so the number of assembly links is at most $k$. 
\end{proof}

By Lemma~\ref{L:assembly-links}, we only need to bound the total number of pairing links.

\section{Node groups}\label{S:node-groups}

The main new idea in our analysis is that of \emph{node groups}.  Node groups are part of the analysis only; the algorithm does not maintain them.  Each group consists of a list of nodes that are consecutive children of some temporary node or are consecutive roots on the root list during the pairing pass of a delete-min. Each child of a temporary node and each root on the root list in the middle of a delete-min is in exactly one group or not in a group.    Children of permanent nodes and roots not in the middle of a delete-min are not in a group.  We call a node that is not in a group \emph{groupless}.  

Groups change according to the following rules:

\begin{itemize}

\item [] Rule 1 (group creation): When a temporary node $v$ is inserted into a heap or is the new root of a heap at the end of a delete-min, a new, empty, leftmost group of children of $v$ is created.

\item [] Rule 2 (group transformation): Consider a delete-min that deletes the root $u$ of a heap and replaces it by a new root $v$.   Deletion of $u$ at the beginning of the delete-min converts each group of children of $u$ into a group of roots.  These groups retain the order they had as groups of children of $u$.  The assembly phase of the delete-min transforms each group of remaining roots into a group of children of $v$.  These new groups of children of $v$ retain the order they had as groups of roots, and they are to the left of all other groups of children of $v$.    

\item [] Rule 3 (group deletion): An empty group is deleted unless it is the leftmost group of children of some node. 

\item [] Rule 4 (node removal): A node $v$ in a group is removed from its group when it loses a non-assembly link, or loses an assembly link to a permanent node, or has its key changed by a decrease-key, or is found to be the new root during the pairing pass of a delete-min. 

\item [] Rule 5 (node addition): When a node $w$ loses a non-assembly link $vw$ to a temporary node $v$, $w$ is added to the leftmost group of children of $v$.

\end{itemize}

At most two groups are created per temporary node, one when the node is inserted, one when it is deleted.  In a list of roots or children, the groups are in the same order as the nodes they contain, and all the groupless nodes are to the left of all the nodes in groups.

Each group is created as the new leftmost group of children of some node, say $x$.  We call $x$ the \emph{initial parent} of the group.  A group retains its initial parent throughout its existence, even when it becomes a group of roots or a group of children of a different node.  This can happen many times.

A group $G$ goes through two phases during its existence.  The first is from the time $G$ is created as the leftmost group of children of its initial parent until it becomes a group of roots or a non-leftmost group of children as the result of a delete-min.  The second is from the time $G$ becomes a group of roots or a non-leftmost group of children until $G$ is deleted.  We say $G$ is \emph{growing} during the first phase and \emph{shrinking} during the second phase.  A delete-min converts at most two groups from growing to shrinking: the leftmost group of children of the node deleted by the delete-min, and the old leftmost group of children of the new root.  A group can only acquire a new node while it is growing, because once it is shrinking, it is not the leftmost group of children of some node, so it can only lose nodes.  The \emph{heap containing a group} $G$ is the heap containing the nodes in $G$; or, if $G$ is empty and is the leftmost group of children of a temporary node $x$, the heap containing $x$. 

The rightmost node in each group is special.  We call a pairing link \emph{outer} if it is won or lost by the rightmost node of a group, \emph{inner} otherwise.  Each outer link is between nodes in two different groups, or between the two rightmost nodes in a group.  Each inner link is between nodes in the same group, neither the rightmost.

We say a pairing link $vw$ is  \emph{of group} $G$ if $G$ contains $w$ before the link.  The link removes $w$ from $G$.  A pairing pass is a pass \emph{of group} $G$ if it does at least one link of $G$.  A link of $G$ is a \emph{pass}-$i$ link of $G$ if it is done during the $i^\mathrm{th}$ pairing pass of $G$.

Let $uv$ be a real inner pairing link.  We call $uv$ a \emph{winner link} if, after $uv$ is cut by the delete-min that deletes $u$, $v$ wins a real inner pass-2 link before $v$ is removed from its current group; that is, before $v$ loses a pairing link, or a decrease-key cuts $v$ from its current parent, or $v$ is the new root at the end of a delete-min.

We conclude this section by proving a lemma that allows us to do time-shifting when stating logarithmic bounds on the number of nodes in a heap.  Let $v$ be a temporary node.  We denote by $N_v$ the maximum number of temporary nodes in the heap containing $v$ while $v$ is in a heap.  Let $G$ be a group.  We define the \emph{span} of $G$ as follows: If $G$ contains a temporary node when it becomes shrinking, its span is the interval of time from its creation until just before it loses a temporary node for the last time; otherwise, its span is empty.  If the span of $G$ is non-empty, then at any time during its span the heap containing $G$ contains at least one temporary node, since while it is growing it must acquire a temporary node that it retains until it becomes shrinking, and while $G$ is growing the heap containing $G$ contains at least the initial parent of $G$.  We denote by $N_G$ the maximum number of temporary nodes in the heap containing $G$ during its span; if the span of $G$ is empty, $N_G = 0$.

\begin{lemma}\label{L:delete-time-shift}
In a sequence of pure pairing heap operations, (i) the sum of $\lg N_v$ over all temporary nodes $v$ is at most $\lg n + 1.5$ per delete-min; (ii) the sum of $\lg N_G$ over all groups $G$ is at most $3\lg n + 3$ per delete-min. 
\end{lemma}
\begin{proof}
To prove (i), we define the \emph{score} of a temporary node $v$ to be $\lg n_v$, where $n_v$ is the number of temporary nodes in the heap containing $v$.  The score of a node can change as the number of temporary nodes in the heap containing it changes.  The maximum score of $v$ is $\lg N_v$.  To bound this maximum score we study changes in scores.    

The score of a temporary node $v$ is zero before it is inserted into a heap and zero once it is deleted from a heap, so the sum of the increases in the score of $v$ equals the negative of the sum of the decreases in the score of $v$.  The score of a node $v$ decreases only when a delete-min deletes a node $u$ from the heap containing $v$.  When $u$ is deleted from the heap, the score of $u$ decreases from $\lg n$ to $0$, and the score of each of the other $n-1$ temporary nodes in the heap decreases by $\lg n - \lg(n-1) \leq (\lg e)/(n-1) < 1.5/(n-1)$.  The sum of the decreases in node scores caused by the delete-min is thus at most $\lg n + 1.5$.  Summing over the delete-mins, the sum of the decreases in node scores over a sequence of heap operations is at most $\lg n + 1.5$ per delete-min.  This is also an upper bound on the sum of the increases in node scores.  The sum of the increases in the node scores of a node $v$ is at least $\lg N_v$.  This gives (i).

We prove (ii) similarly, using (i).  If $G$ is a group with non-empty span, we define its \emph{score} during its span to be $\lg n_G$, where $n_G$ is the number of temporary nodes in the heap containing $G$.

We bound the maximum score of a group by separately bounding its initial score and the sum of later increases in its score.  There are two cases, depending on whether a group is created during an insertion or during a deletion.  Let $G$ be a group with non-empty span, and let $y$ be the last temporary node deleted from $G$.  Suppose $G$ is created during an insertion of a temporary node $x$.  Group $G$ is created as soon as the single-node heap containing $x$ is created, so the initial score of $G$ is $0$.  We charge any increase in the score of $G$ while $G$ is growing to the corresponding increase in the score of $x$.  We charge any increase in the score of $G$ while $G$ is shrinking to the corresponding increase in the score of $y$.  Suppose on the other hand that $G$ is created by a delete-min that deletes node $u$.  The initial score of $G$ is at most $\lg n$, where $n$ is the number of nodes in the heap containing $u$ just before $u$ is deleted.  We charge this initial score to the delete-min of $u$.  Let $v$ be the initial parent of $G$.  We charge any increase in the score of $G$ while it is growing to the corresponding increase in the score of $v$, and we charge any increase in the score of $G$ while it is shrinking to the corresponding increase in the score of $y$.

Since a node can be the parent of at most one growing group at a time, and since a node can be in at most one group at a time, (i) implies that the sum of all the increases in group scores is at most $2\lg n + 3$ per delete-min.  Combining this with the $\lg n$ charge per delete-min for the initial score of groups created by delete-mins gives (ii).
\end{proof}

Now we are ready to bound the number of pairing links.  Our analysis consists of two parts.  First we bound the number of pairing links as a function of winner links and heap operations.  Then we bound the number of winner links.  The first part is straightforward.  The second part adapts and extends the analysis of multipass pairing heaps.

\section{Pairing links}\label{S:pairing-links}

In this section we bound the number of pairing links as a function of winner links and heap operations. To start, we study the effect of pairing passes on a group.  This gives us a bound on the number of outer links.  Then we bound the number of inner links.

\begin{lemma}~\label{L:group-pairing-passes}
Let $G$ be a group.  Let $n_i$ and $k_i$ be the number of nodes in $G$ just before pass $i$ of $G$ and the number of inner pass-$i$ links of $G$, respectively.  (i) For all $i\geq 1$, $n_i/2-3/2\leq k_i\leq n_i/2-1/2$. (ii) If $i$ and $i+1$ are passes of $G$, then $n_{i+1}\leq n_i/2+1$. (iii) If $i$ and $i+1$ are passes of $G$, then $k_{i+1} \leq k_i/2 +1/2$. (iv) If $i$ is a pass of $G$, then $k_i -1 \leq (k_1-1)/2^{i-1}$. (v) The number of inner links of $G$ is at most $2k_1+1$.  (vi) If $i$ is a pass of $G$, the number of passes of $G$ after pass-$i$ is at most $\lg(n_i-2)+2$.  (vii) The total number of passes of $G$ is at most $\lg(n_1-2)+3$. 
\end{lemma}

\begin{proof}
If $n_i$ is odd, $k_i=n_i/2 - 1/2$ or $k_i=n_i/2-3/2$, depending on whether the leftmost member of $G$ participates in a pass-$i$ inner link of $G$ or not.  In either case there are at least $n_i/2-1/2$ pass-$i$ links of $G$, each of which reduces the number of nodes of $G$ by at least $1$.  Thus (i) and (ii) hold if $n_i$ is odd.  If $n_i$ is even, $k_i=n_i/2-1$, whether or not the leftmost member of $G$ participates in a pass-$i$ inner link.  Thus (i) and (ii) hold if $n_i$ is even.

If $n_i$ is odd, $k_i \geq n_i/2 -3/2$ and $n_{i+1} \leq n_i/2 +1/2$ by the proof of (i) and (ii), and $k_{i+1}\leq n_{i+1}/2 -1\xxyedit{/2}$ by (i).  These inequalities combine to give $k_{i+1}\leq k_i/2+1/2$, so (iii) holds if $n_i$ is odd.  If $n_i$ is even, $k_i = n_i/2-1$ and $n_{i+1} \leq n_i/2 +1$ by the proof of (i) and (ii), and $k_{i+1}\leq n_{i+1}/2 -1/2$ by (i). These inequalities combine to give $k_{i+1}\leq k_i/2+1/2$, so (iii) holds if $n_i$ is even.

By (iii), if $i$ and $i+1$ are passes of $G$, then $k_{i+1}-1 \leq (k_i-1)/2$.  Induction on $i$ gives (iv).

If $n_1$ is odd, $k_1\geq n_1/2-3/2$ and $n_2 \leq n_1/2+1/2$ by the proof of (i) and (ii).  The number of inner links of pass greater than $1$ is at most $n_2-1$.  These inequalities combine to give (v) if $n_1$ is odd.  If $n_1$ is even, $k_1=n_1/2-1$ and $n_2\leq n_1/2+1$ by the proof of (i) and (ii), and the number of inner links of pass greater than $1$ is at most $n_2-1$. These inequalities combine to give (v) if $n_1$ is even. 

Suppose pass $i$ is not the last pass of $G$.  By (ii), $n_{i+1}-2 \leq (n_i-2)/2$.  By induction on the number of passes, $n_{i+j}-2 \leq (n_i-2)/2^j$ for pass $i+j$ of $G$ with $j\geq 1$. Hence after at most $\lg(n_i-2)+1$ passes of $G$ following pass $i$ of $G$, $G$ contains at most one node.  Once $G$ contains one node, there is at most one more pass of $G$.  This gives (vi) and (vii).
\end{proof}

Lemma~\ref{L:group-pairing-passes}(vi) gives us a bound on the number of outer links, via the following lemma:

\begin{lemma}\label{L:pairing-passes}
In a sequence of pure pairing heap operations, the total number of pairing passes summed over all groups is $\OO(1)$ per insertion plus $\OO(1)$ per decrease-key plus at most $3\lg n$ per delete-min.  
\end{lemma}
\begin{proof}
Let $G$ be a group.  We account for its pairing passes as follows.  If some link of $G$ during the pass is won by a permanent node, this link is either an f-link or a k-link.  We charge the pass to the insertion of the permanent node that loses the f-link or to the decrease-key that cuts the k-link, respectively.  If no link of $G$ during the pass is won by a permanent node, and $G$ is the rightmost group of roots just before the pass, we charge the pass to the insertion of the node that is deleted by the delete-min that does the pairing pass.  The remaining possibility is that every node of $G$ participates in a link during the pairing pass, and each of these links is won by a temporary node.  In this case we charge the pass to $G$ itself.  Group $G$ contains only temporary nodes after such a pass, so once this case occurs, it repeats.  By Lemma~\ref{L:group-pairing-passes}(vi), the number of times this case can occur is at most $\lg(N_G-2)+3$.   Summing the charges and applying Lemma~\ref{L:delete-time-shift}(ii) to the charges to each group gives the lemma.      
\end{proof}

\begin{theorem}\label{T:outer-link-bound}
In a sequence of pure pairing heap operations, the number of outer links is $\OO(1)$ per insertion plus $\OO(1)$ per decrease-key plus at most $6\lg n$ per delete-min. 
\end{theorem}
\begin{proof}
There are at most two outer links per pass per group.  The theorem therefore follows immediately from Lemma~\ref{L:pairing-passes}.    
\end{proof}

Theorem~\ref{T:outer-link-bound} gives us a bound on the number of outer links.  All that remains is to bound the number of inner links.  By Lemma~\ref{L:group-pairing-passes}(v), the total number of inner links is at most twice the number of pass-$1$ inner links plus one per group, so it is enough to bound the number of pass-$1$ inner links.  We proceed to do this.

\begin{lemma}\label{L:pass-1-inner-links}
The number of pass-$1$ inner links in a sequence of pure pairing heap operations is at most two per real pass-2 inner link plus seven per insertion plus two per decrease-key.
\end{lemma}

\begin{proof}
Let $G$ be a group.  We bound the number of winners of pass-$1$ inner links of $G$.  At the end of pass-$1$ of $G$, $G$ contains the winners of the pass-$1$ inner links of $G$ plus at most two additional nodes, the leftmost and rightmost nodes in the group.  We consider what happens to these winners between the end of pass $1$ of $G$ and the end of pass $2$ of $G$.

Consider the delete-min that does pass $1$ of $G$.  If the new root resulting from this delete-min is a permanent node, then all the assembly links done by the delete-min, including one per pass-$1$ pairing link, are k-links or f-links.  In this case we charge each each winner of a pass-$1$ assembly link to the decrease-key that cuts the k-link of the insertion of the loser of the f-link, respectively.

Suppose on the other hand that the new root resulting from the delete-min is a temporary node.  In this case each pass-$1$ winner \xxydelete{except at most one, the new root,} remains in $G$ until pass $2$ of $G$, unless it is the new root, or it is removed by a decrease-key operation.  In the former case we charge the winner to the delete-min that does pass $1$ of $G$; in the latter we charge the winner to the decrease-key that removes it from $G$.  

Among the pass-$1$ winners still in $G$ at the beginning of pass-$2$, at most three do not participate in an inner pass-$2$ link between two pass-$1$ winners.  We charge these to $G$.  Let $vw$ be an inner pass-$2$ link of two pass-$1$ winners.  If $v$ or $w$ is permanent, the pass-$1$ link it won must be an f-link or a k-link.  We charge $v$ and $w$ to the loser of the f-link or to the decrease-key that cuts the k-link, respectively.  If $v$ and $w$ are temporary but $vw$ is a k-link, we charge $v$ and $w$ to the decrease-key that cuts this k-link.  The remaining possibility is that $vw$ is a real inner pass-$2$ link of $G$.  In this case we charge $v$ and $w$ to $vw$.

In total, each decrease-key and each real pass-$2$ link is charged for at most two pass-$1$ links, each permanent node is charged for at most one pass-$1$ link, each delete-min is charged for at most one pass-$1$ link, and each group is charged for at most three pass-$1$ links.  There are at most two groups per temporary node and exactly one delete-min per temporary node.  The lemma follows.
\end{proof}

\begin{lemma}\label{L:real-pass-2-links}
The number of real pass-$2$ inner links is at most the number of winner links plus $6\lg n$ per delete-min plus $\OO(1)$ per insertion plus $\OO(1)$ per decrease-key.
\end{lemma}
\begin{proof}
Let $vw$ be a real pass-$2$ inner link.  Let $u$ be the initial parent of the group containing $v$ and $w$ when this link is done.  A real link $uv$ was cut when $u$ was deleted by a delete-min.  If $uv$ is an inner link, it is a winner link, since $v$ remains in the same group as $w$ until it wins the link $vw$.  If $uv$ is not an inner link, it is an outer link, a meld link, or a decrease-key link.  Link $vw$ is uniquely determined by link $uv$, since $v$ can only win one pass-$2$ pairing link while it remains in the same group as $w$.  It follows that the number of real pass-$2$ inner links is at most the number of winner links plus the number of outer links plus twice the number of insertions plus the number of decrease-key operations.  Applying Theorem~\ref{T:outer-link-bound} to bound the number of outer links gives the lemma.    
\end{proof}

Let $I$ and $W$, respectively, be the number of pass-$1$ inner links and the number of winner links in a sequence of pure pairing heap operations.  Lemmas~\ref{L:pass-1-inner-links} and~\ref{L:real-pass-2-links} combine to give us the following result:

\begin{lemma}\label{L:winner-links}
$I$ is at most $2W$ plus $\OO(1)$ per insertion plus $\OO(1)$ per decrease-key plus $\OO(\lg n)$ per delete-min.
\end{lemma}

We restate this lemma as the following corollary:

\begin{corollary}\label{C:many-winner-links}
$I$ is at most $6(W-I/3)$ plus $\OO(1)$ per insertion plus $\OO(1)$ per decrease-key plus $\OO(\lg n)$ per delete-min.
\end{corollary}

By Corollary~\ref{C:many-winner-links}, our desired bound on the efficiency of pure pairing heaps holds if $W < I/3$.  Thus our remaining task is to bound $W$, assuming that $W \geq I/3$.

\section{Node ranks}\label{S:node-ranks}
In this and the next two sections, we do the hard part of the analysis, that of bounding the number of winner links, assuming that the number of winner links is at least $1/3$ the number of pass-$1$ inner links.  We apply the multipass-pairing-heap analysis of~\cite{SinnamonT25} to each group separately.  This requires some changes in their definitions, as well as some extra calculations to account for the effect of outer links.

In this section we define and study the first key concept in the analysis, that of \emph{node ranks} and their generalization, $b$-\emph{reduced node ranks}.  These are part of the analysis only: The algorithm does not use or maintain them.  We use node ranks to quantify the effect of links on the data structure.  We bound the number of winner links by showing that, except for a bounded number of exceptions, (i) each winner link can be charged to a set of changes in the magnitudes of certain node ranks totaling at least a positive constant, and (ii) the sum of the magnitudes of all node-rank changes is appropriately bounded.  

The rank of a temporary node, and its $b$-reduced rank for a non-negative integer $b$, are functions of the \emph{size} and the \emph{mass} of the node.  We define these values as follows: 

\begin{itemize}

\item [] The \emph{size} $v.s$ of a temporary node $v$ is the number of temporary nodes that are descendants of $v$ and that are connected to $v$ by a path of real links.  This includes $v$ itself.

\item [] The definition of the \emph{mass} $v.m$ of a temporary node $v$ depends on the state of $v$. If $v$ is a phantom child or a root not in a group, its mass is its size.  If $v$ is a root in the middle of a delete-min, its mass is one plus the sum of its size and those of the temporary nodes that are to its right in its group, not including the rightmost node in the group. If $v$ is a real child of node $u$, its mass is one plus the sum of its size and those of the real children of $u$ that are to its right in its group, not including the rightmost node in the group.  In particular, the mass of a node that is one of the two rightmost nodes in a group is one plus its size. There is one momentary exception to this definition of mass: If a node $w$ loses a real right inner or real left outer link $vw$, just before the link the mass of $w$ increases to equal the mass of $v$ after the link.  We call this a \emph{step-up} in the mass of $w$.  If the link is a real right inner link, the step-up increases the mass of $w$ by the size of $v$ before the link.  If the link is a real left outer link, the step-up increases the mass of $w$ by the mass of $v$ before the link.

\item [] The \emph{node rank} $v.r$ of a temporary node $v$ in a heap is $\lg\lg(v.m/v.s)$; the rank of a temporary node about to be inserted into a heap or just deleted from a heap is $0$.

\item [] Let $b$ be a non-negative integer.  The $b$-\emph{reduced node rank} $v.r(b)$ of a temporary node $v$ in a heap is $\max\{v.r - b, 0\} = \lg(\lg (v.m/v.s)/2^b)$.  In particular, the $0$-reduced rank of $v$ is its rank.
\end{itemize}  

These definitions are the ones used in~\cite{SinnamonT25} with two changes: (i) The mass of a node includes the sizes only of other nodes to its right in its group, excluding the rightmost, rather than of those of all nodes to its right in the list of siblings or roots containing it; and (ii) real left outer links as well as real right inner links cause a step-up in the mass of the loser of the link just before the link is done.  The node-mass restriction limits spillovers of mass from one group to another.  The extension of step-ups to outer left links guarantees that link ranks, which we shall define and study in the next section, are non-negative.  In the analysis of multipass pairing heaps, there are no outer links, so there is no need for outer-link step-ups.  

The sizes, masses, ranks, and $b$-reduced ranks of nodes can change over time as links and cuts of links occur.  The definitions imply that (i) the size and mass of a temporary node $v$ are positive integers, (ii) while $v$ is in a heap, the size of $v$ is non-decreasing, (iii) $v.m \geq v.s$, (iv) the $b$-reduced rank of $v$ is non-negative, and (v) the $b$-reduced rank of $v$ is zero if $v$ is a root not in the middle of a delete-min or is one of the two rightmost nodes in a group.

We develop some results about node ranks, $b$-reduced node ranks, and how much they can change.  Most of these are direct analogues of results in the analysis of multipass pairing heaps~\cite{SinnamonT25}, and the proofs are almost the same.  The exceptions are those affected by outer links.  We provide proofs except in the cases where the proofs are identical or almost identical to those in~\cite{SinnamonT25}.

\begin{lemma}\label{L:small-rank-nodes}
In a sequence of pure pairing heap operations, for any fixed $i$, the number of real inner pass-$i$ links in which the winner or loser has rank at most $1$ just before the link is at most $(15/2)\lg n + 19/2$ per delete-min.
\end{lemma}
\begin{proof}
This lemma is the analogue for pure pairing heaps of Corollary 7.7 in~\cite{SinnamonT25}; its proof is the same except for the use of Lemma~\ref{L:delete-time-shift}.
\end{proof}

\begin{lemma}\label{L:decrease-key-rank-increases}
In a pure pairing heap, a real decrease-key link $vw$ increases the $b$-reduced rank only of $w$, by at most $\lg((\lg n)/2^b)$.
\end{lemma}
\begin{proof}
This lemma is the analogue for pure pairing heaps of Lemma 7.20 in~\cite{SinnamonT25}, and its proof is the same. 
\end{proof}

\begin{lemma}\label{L:meld-rank-increases}
In a sequence of pure pairing heap operations, the sum of increases in $b$-reduced ranks caused by real meld links is at most $\lg((\lg n)/2^b)$ per decrease-key plus $\lg((\lg n)/2^b)+3/2^b$ per delete-min. 
\end{lemma}
\begin{proof}
This lemma is the analogue for pure pairing heaps of Lemma 7.21 in~\cite{SinnamonT25}, and its proof is the same. 
\end{proof}

\begin{lemma}\label{L:inner-link-rank-increases}
Let $G$ be a group, let $a$ be any non-negative integer, and let $k_i$ be the number of pass-$i$ inner links of $G$. The sum of the $b$-reduced rank increases caused by the pass-$i$ inner links of $G$ is at most $(2/2^{a+b})\lg N_G +ak_i$.
\end{lemma}
\begin{proof}
This lemma is the analogue for pure pairing heaps of Lemma 7.23 in~\cite{SinnamonT25}, and its proof is essentially the same. 
\end{proof}

Recall from Section~\ref{S:pairing-links} that $I$ is the total number of pass-$1$ inner links in a sequence of pure pairing heap operations.

\begin{lemma}\label{L:delete-rank-increases}
(i) Let $G$ be a group, let $c$ be any non-negative integer, and let $k_1$ be the number of pass-$1$ inner links of $G$.  The sum of the $b$-reduced rank increases caused by the links of $G$ is at most $((6b+6c+22)/2^b)\lg N_G + 8k_1/2^{b+c}$.  (ii) The sum of the $b$-reduced rank increases caused by all the pairing links in a sequence of pure pairing heap operations is at most $((6b+6c+22)/2^b)(3\lg n +3)$ per delete-min plus $8I/2^{b+c}$.   
\end{lemma}
\begin{proof}
We bound the rank increases caused by outer links and sum these increases and those of the inner links, which are bounded by Lemma~\ref{L:inner-link-rank-increases}.

Let $a$ be any non-negative integer.  In a pairing pass of $G$, there are at most two outer links of $G$. Let $vw$ be such a link.  Including the step-up if $vw$ is a left link, the link increases the $b$-reduced rank only of $w$, by at most $\lg((\lg N_G)/2^b)\leq (2/2^{a+b})\lg N_G + a$, where the inequality follows from Lemma 7.19  in~\cite{SinnamonT25} and can be proved by a simple calculation.  Combining this bound with that of Lemma~\ref{L:inner-link-rank-increases}, we find that the sum of the $b$-reduced rank increases caused by the links of pairing pass-$i$ of $G$ is at most $(6/2^{a+b})\lg N_G + a(k_i+2)$, where  $k_i$ is the number of pass-$i$ inner links of $G$ done by the pass. We set $a = a_i = \max\{i-b-c-1, 0\}$ for pass-$i$ and sum this bound over $i$ such that $k_i\geq 2$. %\xxycomment{Without this $-1$ here, then the formula below should contain the term $4jk_1/2^{j + b + c {\color{red}\bf {}- 1}}$, and it needs to change a lot of terms.} 
In doing the sum, we use the substitution $j=i-b-c\xxyedit{{}-1}$ for $i>b+c$, and we use the inequalities $k_i-1\leq (k_1-1)/2^{i-1}$  given by Lemma~\ref{L:group-pairing-passes}(iv) and $k_i+2 \leq 4(k_i-1)$,  which follows from the assumption that $k_i\geq 2$. Note that when $i \geq \lg k_1 +1$, $k_i\leq 1 + (k_1-1)/2^{i-1} < 2$. The sum is at most
\[((6b+6c)/2^b)\lg N_G + \sum_{j=0}^\infty ((6/2^{j + b})\lg N_G + 4jk_1/2^{j+b+c}) \leq ((6b+6c+12)/2^b)\lg N_G + 8k_1/2^{b+c} \]

A pass of $G$ does fewer than two inner links only if $G$ contains at most five nodes just before the pass.  The total number of links of $G$ during these passes is at most five.  Adding a bound of $(10/2^b)\lg N_G$ on the $b$-reduced rank increases caused by these links gives (i). This and Lemma~\ref{L:delete-time-shift} give (ii).       
\end{proof}

\begin{corollary}\label{C:outer-link-rank-bound}
(i) Let $G$ be a group.  The product of $\lg\lg N_G$ times the number of outer links of $G$ is  $\OO(\lg N_G)$ plus at most $k_1/288$, where $k_1$ is the number of pass-$1$ inner links of $G$.  (ii) The sum of this product over all groups in a sequence of pure pairing heap operations is  $\OO(\lg n)$ per delete-min plus at most $I/288$.
\end{corollary}
\begin{proof}
The corollary follows from Lemma~\ref{L:delete-rank-increases} by setting $b=0$ and $c$ sufficiently large and observing that the bounds in the proof of the lemma include a term of $\lg\lg N_G$ for each outer link of a group $G$.
\end{proof}

\begin{theorem}\label{T:reduced-node-rank-changes} 
Let $c$ be any non-negative integer.  In any sequence of pure pairing heap operations, the sum of the increases in $b$-reduced node ranks is $\OO(\lg((\lg n)/2^b))$ per decrease-key plus $\OO(((b+c+1)\lg n + b+c+1)/2^b)$ per delete-min plus at most $8I/2^{b+c}$, where the constants in the big $\OO$ notation are independent of $b$, $c$, and the sequence of heap operations.  The same bound holds for the sum of the magnitudes of the decreases in $b$-reduced ranks.  
\end{theorem}

\begin{proof}
This is the analogue for pure pairing heaps of Theorem 7.13 in~\cite{SinnamonT25}.  The first half of the theorem follows immediately from Lemmas~\ref{L:meld-rank-increases},~\ref{L:decrease-key-rank-increases}, and~\ref{L:delete-rank-increases}.  The second half follows from the observation that the sum of the magnitudes of the decreases in $b$-reduced node ranks equals the sum of the increases in $b$-reduced node ranks. 
\end{proof}

\begin{theorem}\label{T:node-rank-changes}  
In any sequence of pure pairing heap operations, the sum of the node-rank increases is $\OO(\lg\lg n)$ per decrease-key plus $\OO(\lg n)$ per delete-min plus $I/72$.  The same bound holds for the sum of the magnitudes of the node-rank decreases.   
\end{theorem}
\begin{proof}
This theorem is the analogue for pure pairing heaps of Theorem 7.10 in~\cite{SinnamonT25}.  The theorem follows from Theorem~\ref{T:reduced-node-rank-changes} by setting $b=0$ and setting $c$ large enough.
\end{proof}

\section{Link ranks}\label{S:link-ranks}

The second key concept in the analysis of inner links is that of \emph{link ranks}.  Link ranks connect winner links with changes in node ranks.  This allows us to use the results in Section~\ref{S:node-ranks} to bound the number of winner links and hence the number of inner links.  In this section we define and study link ranks.

We use the following notation, an extension of that of~\cite{SinnamonT25}, for the rank of a node at certain times: If $vw$ is a link both of whose nodes are temporary, $v.r_{vw}$, $w.r_{vw}$, $v.r'_{vw}$, and $w.r'_{vw}$ denote the rank of $v$ and the rank of $w$ just before and just after $vw$ is done, respectively; if $vw$ is a real right inner or real left outer link, $w.r_{vw}$ is the rank of $w$ just after its mass steps up just before $vw$ is done.  We use the analogous notation for $b$-reduced ranks just before and just after a link.  We also denote by $n_{vw}$ the number of temporary nodes in the heap containing $v$ and $w$ just after link $vw$ is done. 

If $vw$ is a real link, its \emph{link rank} $vw.r$ is $w.r_{vw}-v.r'_{vw}$.

\begin{lemma}\label{L:non-negative-link-ranks} 
The link rank of a real %\xxyedit{non-assembly} 
link is non-negative.
\end{lemma}
\begin{proof}
This lemma is the analogue for pure pairing heaps of Lemma 7.5 in~\cite{SinnamonT25}, and its proof is the same for meld links, decrease-key links, and inner links.  We prove it for real outer links.

Let $vw$ be a real outer link.  If $w$ is to the right of $v$ before the link, the rank of $v$ is $0$ after the link, making the lemma true.  If $w$ is to the left of $v$ before the link, the step-up in the mass of $w$ before the link sets $w.m_{vw} \geq v.m_{vw}'$.  Hence $w.m_{vw}/w.s_{vw} \geq v.m_{vw}'/(v.s_{vw}+w.s_{vw})=v.m'_{vw}/v.s'_{vw}$, making the lemma true in this case also.  \end{proof}

The following lemma is the key to our analysis, as it is to the analysis of~\cite{SinnamonT25}.

\begin{lemma}\label{L:small-link-rank} 
Let $vw$ be a real inner pairing link with link rank less than $1$.  If both $v$ and $w$ have positive rank just after $vw$ is done, then doing $vw$ decreases the rank of $w$ by at least $1$. 
\end{lemma}
\begin{proof}
This is the analogue for pure pairing heaps of Lemma 7.8 in~\cite{SinnamonT25}, and its proof is the same.
\end{proof}

The following corollary is a more convenient form of this lemma:

\begin{corollary}\label{C:small-link-rank}
Let $vw$ be a real inner pairing link with link rank less than $1$.  If both $v$ and $w$ have rank at least $1$ just before $vw$ is done, then doing $vw$ decreases the rank of at least one of $v$ and $w$ by at least $1$. 
\end{corollary}
\begin{proof}
Suppose doing $vw$ does not decrease the rank of either $v$ or $w$ by at least one.  Then both $v$ and $w$ have positive rank just after the link.  But then by Lemma~\ref{L:small-link-rank}, doing $vw$ decreases the rank of $w$ by at least $1$.
\end{proof}

Corollary~\ref{C:small-link-rank} plus the results in Section~\ref{S:node-ranks} will allow us to derive a bound on the number of early-pass, small-rank winner links.  We bound the number of large-rank winner links by bounding the sum of the ranks of all the winner links.  These results, together with the assumption $W \geq I/3$, where $W$ is the total number of winner links and $I$ is the total number of pass-$1$ inner links, give us a bound on $W$, and hence on the number of pass-$1$ inner links, completing our analysis.

To obtain a bound on the sum of the link ranks of the winner links, we use a crucial equality from~\cite{SinnamonT25}, the \emph{link-rank} equality.  Let $vw$ be a real inner link.  Let $u$ be the initial parent of the group containing $v$ and $w$ when $vw$ is done.  The delete-min that deletes $u$ cuts real links $uv$ and $uw$.  The link-rank equality for $vw$ is:

\[v.r_{uv}-u.r'_{uv}=(v.r'_{vw}-u.r'_{uw})+(v.r_{uv}-v.r'_{vw})+(u.r'_{uw}-u.r'_{uv})\] 

The equality follows by canceling equal terms on the right side. The left side of the link-rank equality is the link rank of $uv$.  We call the three terms in parentheses on the right side the $w$-\emph{shift}, $v$-\emph{shift}, and $u$-\emph{shift} of $vw$, respectively.  

If $uv$ is a winner link, there is a pass-$2$ real inner link $vw$ such that $u$ is the initial parent of the group containing $v$ and $w$ when $vw$ is done.  Thus, if we sum the link rank equality over all real inner links $vw$, the left side includes the link rank of each winner link at least once.  This allows us to bound the sum of the link ranks of the winner links by bounding the sums of the $w$-shifts, the $v$-shifts, and the $u$-shifts of all the real inner links.  This is what we shall do.

\begin{lemma}\label{L:w-shifts}  
In a sequence of pure pairing heap operations, the sum of the $w$-shifts of all the real inner links is $\OO(\lg\lg n)$ per decrease-key plus $\OO(\lg n)$ per delete-min plus $I/288$. 
\end{lemma}
\begin{proof}
This is the analogue for pure pairing heaps of Lemma 6.2 in~\cite{SinnamonT25}.  Its proof is the same except that we need to account for the effect of outer links.

Let $w$ be a temporary node.  Consider the sequence of non-assembly links lost by $w$.  Let $v_0w, v_1w,\dots, v_kw$ be a maximal consecutive subsequence of these links such that each one except $v_0w$ is a real inner link.  For $1 \leq i \leq k$, the $w$-shift of $v_iw$ resulting from doing $v_iw$ and earlier cutting $v_{i-1}w$ is $v_i.r'_{v_iw}-v_{i-1}.r'_{v_{i-1}w}$.  The sum of these $w$-shifts telescopes to $v_k.r'_{v_kw}-v_0.r'_{v_0w} \leq v_k.r'_{v_kw}\leq \lg\lg n_{v_kw}$.  It follows that the sum of the $w$-shifts of the real inner links over a sequence of pure pairing heap operations is at most the sum of $\lg\lg n_{vw}$ over all real inner links $vw$ such that the next non-assembly link lost by $w$ after $vw$, if any, is not a real inner link.  We bound this sum by charging each term to an appropriate event.

Let $vw$ be a real inner link such that the next non-assembly link lost by $w$ after $vw$, if any, is not a real inner link.  \xxydelete{Let $u$ be the initial parent of $v$ and $w$.}  The link $vw$ removes $w$ from the group previously containing $v$ and $w$ and adds $w$ to the growing group of children of $v$, say group $G$.  We consider the history of $w$ in its new group $G$.  Since $w$ is not deleted from $G$ by losing a real inner link, there are only three ways in which $w$ can be deleted from $G$: by losing a real outer link, by losing a k-link, or by becoming the new root at the end of a delete-min.  In each of the first two cases, we charge $\lg\lg n_{vw}$ to the link that removes $w$ from $G$.  In the third case, we charge $\lg\lg n_{vw}$ to the delete-min that makes $w$ the new root.  Each link and each delete-min is charged at most once.

To bound the total charges, we consider each of the three cases separately.  Suppose a real outer link $xw$ of group $G$ is charged for a link $vw$.  Then $G$ has non-empty span, since $xw$ is done while $G$ is shrinking.  Hence $vw$ is done during the span of $G$, which implies that $\lg\lg n_{vw}\leq \lg\lg N_G$.  By Corollary~\ref{C:outer-link-rank-bound}, the total charge to real outer links is $\OO(\lg n)$ per delete-min plus at most $I/288$.

We bound the sum of charges to k-links by shifting some of these charges to decrease-key and delete-min operations, using a method like that in the proof of Lemma~\ref{L:delete-time-shift}. We allocate $\lg\lg n$ credits to each decrease-key and $3$ credits to each delete-min, and show that the total number of credits is enough to pay for the charges to k-links.  

We apply the credits assigned to a decrease-key to partially or fully pay for the charge corresponding to the k-link cut by the decrease-key.  This leaves $3$ credits per delete-min to pay for any unpaid charges to k-links.  The unpaid charge of a k-link $xw$ is $\lg\lg n_{vw}-\lg\lg n$, where $vw$ is the link lost by $w$ preceding its loss of $xw$, and $n$ is the \xxyreplace{size of}{number of temporary nodes in} the heap containing $w$ just before the decrease-key that cuts $xw$.  The unpaid charge is positive only if $n < n_{vw}$.

When a delete-min occurs, it allocates its $3$ extra credits equally to the remaining nodes in the heap, $3/(n-1)$ to each.  From the time $vw$ is done until the time $xw$ is cut, $w$ accumulates at least $3\sum_{i=n}^{n_{vw}-1}1/i \geq 2(\lg n_{vw}-\lg n) \geq \lg\lg n_{vw}-\lg\lg n$ credits, enough to cover the unpaid charge to $xw$.  The time interval during which $w$ accumulates these credits is disjoint from the interval for any other such pair of links $v'w$, $x'w$, so the sum of the allocated credits is an upper bound on the total charges.

We bound the sum of charges to delete-mins by doing a similar shifting.  We allocate $\lg n+1.5$ credits to each delete-min, and show that the total number of credits is enough to pay for the charges to delete-mins.  We denote by $n_x$ the number of \xxyedit{temporary} nodes in the heap containing $x$ just before $x$ is deleted.

We apply $\lg n_x$ of the credits assigned to the delete-min of a node $x$ to partially or fully pay for the charge to this delete-min.  This leaves $1.5$ credits per delete-min to cover any unpaid charges to delete-mins.  Suppose the delete-min of a node $x$ is charged $\lg\lg n_{vw}\leq \lg n_u$, where $u$ is the node whose deletion does link $vw$.  The unpaid charge to the delete-min of $x$ is at most $\lg n_u-\lg n_x$.  This is positive only if $n_u > n_x$.

When a delete-min occurs, it allocates its $1.5$ extra credits equally to the remaining nodes in the heap, $1.5/(n-1)$ to each.  During the interval of the time from the deletion of $u$ until the deletion of $x$, $w$ accumulates at least $\sum_{i=n_x}^{n_u-1} 1.5/i \geq \lg n_u-\lg n_x$ credits, enough to cover the unpaid charge to the delete-min of $x$.  During this interval $w$ is a root only in the middle of a delete-min.  Since $w$ is the new root after $x$ is deleted, the time interval during which $w$ accumulates these extra credits is disjoint from any time interval during which $w$ accumulates extra credits for another delete-min. Thus the sum of the allocated credits is an upper bound on the total charges to nodes.

The lemma follows by combining the bounds on the charges to outer links, k-links, and nodes.
\end{proof}

\begin{lemma}\label{L:v-shifts} 
In a sequence of pure pairing heap operations, the sum of the $v$-shifts of all the real inner links is $\OO(\lg\lg n \cdot \lg\lg\lg n)$ per decrease-key plus $\OO(\lg n)$ per delete-min plus at most %$I/216$.
$I/288$.
\end{lemma}
\begin{proof}
This lemma is the analogue for pure pairing heaps of Lemma 7.14 in~\cite{SinnamonT25}, and the proof is almost the same, except that we must account for groups.  Let $vw$ be a real inner link, and let $u$ be the initial parent of $v$ and $w$.  The $v$-shift of $vw$ is $v.r_{uv} - v.r'_{vw}$, the net decrease in the rank of $v$ from just before $uv$ is done until just after $vw$ is done.  We call this interval the \emph{v-shift interval} of $vw$.  If $vw$ is a pass-$i$ link, $v$ cannot win another pass-$i$ link until losing a link.  Thus the $v$-shift intervals for the pass-$i$ real inner links won by $v$ are disjoint.  This allows us to apply the second half of Theorem~\ref{T:reduced-node-rank-changes} to bound the sum of the $v$-shifts of all pass-$i$ links for a given $i$.  Summing over $i$ gives the lemma.

When summing, we ignore negative $v$-shifts.  If the $v$-shift of $vw$ is non-negative, it is $v.r_{uv} - v.r'_{vw}\leq v.r_{uv}(b)- v.r'_{vw}(b)+ b$ for any non-negative integer $b$, where the inequality follows from Lemma 7.12 in~\cite{SinnamonT25} and can be proved by a simple calculation.  Given a group $G$, let $k_i$ be the number of inner pass-$i$ links of $G$.  By Lemma~\ref{L:group-pairing-passes}~(iv), $k_i\leq k_1/2^{i-1}+1$.  We bound the sum of the non-negative $v$-shifts of real pass-$i$ inner links by using the second half of Theorem~\ref{T:reduced-node-rank-changes} with appropriately chosen $b_i$, and adding the sum of $b_ik_i$ over all groups.  We obtain the lemma by summing the bounds for each $i$ over all $i$ and choosing $c$ large enough.

Let $c$ be an arbitrary non-negative integer.  Let $b_i=0$ for $i \leq c$, $b_i=2^{\lfloor (i-c-1)/2 \rfloor}$ for $i > c$. Given a group $G$, let $k_i$ be the number of inner pass-$i$ links of $G$; $k_i=0$ if there is no pass $i$ of $G$.  Let $j$ be the maximum index such that $k_j \geq 2$; $j=0$ if there is no $k_i \geq 2$.  \xxyreplace{After pass $j$}{Before pass $j+1$}, $G$ contains at most five nodes; \xxyreplace{after pass $j+1$}{before pass $j+2$}, $G$ contains at most three nodes.  It follows that $k_{j+2}\leq k_{j+1}\leq 1$, and $k_i=0$ for $i\geq j + 3$.  For $1 \leq i \leq j$, $k_i \leq 2(k_i-1)\leq k_1/2^{i-2}$.  It follows that 
\[\sum_{i=1}^j {b_ik_i} \leq \sum_{i=1}^j{b_i k_1/2^{i-2}} \leq 4k_1\sum_{i=c+1}^\infty b_i/2^i \leq 8k_1/2^c\].
% There need not exist a j with k_j >= 2. In that case the corresponding part is simply 0, and one can take j = 0 in this case. It is also a generalization that if pass-j does not exist for j >= 1 then k_j = 0. Every reader is supposed to understand it without being explict. Making it explicit will significantly add unnecessary verbosity to the proof.  Bob: I made it explicit.

To this bound we must add $b_{j+1}+b_{j+2}$ to account for the at most two passes with $k_i=1$.  We consider two cases.  If $j>c$, then $4b_j\geq b_{j+1}+b_{j+2}$.  Since $\sum_{i=1}^j {b_ik_i}\geq 2b_j$, $b_{j+1}+b_{j+2}\leq 16k_1/2^c$, so the sum of $b_ik_i$ over all $i$ is at most $24k_1/2^c$.  On the other hand, if $j\leq c$, then $b_{j+1}+b_{j+2}\leq 2$.  Hence in both cases the sum of $b_ik_i$ over all $i$ is at most $24k_1/2^c$ plus $2$.  The sum of this term over all groups is at most $24I/2^c$ plus $4$ per delete-min.

The term in $I$ in the bound of Theorem~\ref{T:reduced-node-rank-changes} is $8I/2^{b_i+c}$.  The sum over $i$ of this term is at most $8(c+2)I/2^c$.  The bound per decrease-key in Theorem~\ref{T:reduced-node-rank-changes} is $\OO(\lg((\lg n)/2^{b_i}))$.  The sum over $i$ of $\lg((\lg n)/2^{b_i})$ is $\OO(\lg\lg n \cdot (\lg\lg\lg n + c+1))$, since $i > 2\lg\lg\lg n+c+2$ implies $2^{b_i} \geq\lg n$.  The bound per delete-min in Theorem~\ref{T:reduced-node-rank-changes} is $\OO(((b_i+c+1)\lg n + b_i+c+1)/2^{b_i})$.  The sum over $i$ of $((b_i+c+1)\lg n + b_i+c+1)/2^{b_i})$ is $\OO((c+1)\xxyedit{^2}(\lg n+1))$.  Combining the bounds and choosing $c$ large enough gives the lemma.      
\end{proof}

\begin{lemma}\label{L:u-shifts}   
In a sequence of pure pairing heap operations, the sum of the $u$-shifts of all the real inner links is $\OO(\lg\lg n \cdot \lg\lg\lg n)$ per decrease-key plus $\OO(\lg n)$ per delete-min plus at most %$I/216$.
$I/288$.
\end{lemma}
\begin{proof}
This lemma is the analogue for pure pairing heaps of Lemma 7.15 in~\cite{SinnamonT25}, and the proof is almost the same, except that we must account for groups.  Let $vw$ be a real inner link, and let $u$ be the initial parent of $v$ and $w$.  The $u$-shift of $vw$ is $u.r'_{uw} - u.r'_{uv}$.  This is at most the magnitude of the change in the rank of $u$ from the time $t_1$ just after the first of  $uw$ and $uv$ is done until the time $t_2$ just after the second of $uw$ and $uv$ is done.  We call the interval $[t_1, t_2]$ the \emph{u-shift interval} of $vw$.  For a given group $G$ and a given pass $i$, the $u$-shift intervals of the real inner pass-$i$ links of $G$ are disjoint.  
It follows that we can use Theorem~\ref{T:reduced-node-rank-changes} to bound the sum of the $u$-shifts in the same way we used it to bound the sum of the $v$-shifts in the proof of Lemma~\ref{L:v-shifts}, except that we must use both halves of the lemma, since both positive and negative changes in node ranks can contribute to $u$-shifts.
\end{proof}

\begin{theorem}\label{T:link-ranks}
In a sequence of pure pairing heap operations, the sum of the link ranks of all winner links is $\OO(\lg\lg n \cdot \lg\lg\lg n)$ per decrease-key plus $\OO(\lg n)$ per delete-min plus at most $I/72$.
\end{theorem}
\begin{proof}
This theorem is the analogue for pure pairing heaps of Theorem 7.16 in~\cite{SinnamonT25}.  It is immediate from Lemmas~\ref{L:w-shifts},~\ref{L:v-shifts}, and~\ref{L:u-shifts}. 
\end{proof}

\section{Winner links}
\label{S:winner-links}
Now we have all the results we need to bound the number of winner links and hence the number of pass-$1$ inner links, completing the analysis.

\begin{theorem}\label{T:inner-link-bound}
In a sequence of pure pairing heap operations, the number of pass-$1$ inner links is $\OO(1)$ per insertion plus $\OO(\lg\lg n\cdot\lg\lg\lg n)$ per decrease-key plus $\OO(\lg n)$ per delete-min.  
\end{theorem}
\begin{proof}
By Corollary~\ref{C:many-winner-links}, the theorem holds if $W < I/3$.  Thus assume $W \geq I/3$. 

Given a group $G$, let $k_i$ be the number of inner pass-$i$ links of $G$.  By Lemma~\ref{L:group-pairing-passes}(iii), $k_3 \leq k_2/2 + 1/2\leq k_1/4 + 1/4+ 1/2 = k_1/4+3/4$.  Since $G$ contains at most $k_3+2$ nodes after pass $3$, the total number of pass-$i$ inner links of $G$ with $i\geq 4$ is at most $k_3+1\leq k_1/4+2$. Summing over groups, the total number of pass-$i$ inner links with $i\geq 4$ is at most $I/4$ plus $4$ per delete-min.  Given that $W \geq I/3$, this implies that at least $I/12$ inner links minus $4$ per delete-min are winner links of pass $1$, $2$, or $3$.

Suppose that at least $I/36$ winner links minus $4$ per delete-min have a link rank of at least $1$.  By Theorem~\ref{T:link-ranks}, the number of winner links with link rank at least $1$ is $\OO(\lg\lg n \cdot \lg\lg\lg n)$ per decrease-key plus $\OO(\lg n)$ per delete-min plus at most $I/72$.  Since $1/36> 1/72$, $I$ is $\OO(\lg\lg n \cdot \lg\lg\lg n)$ per decrease-key plus $\OO(\lg n)$ per delete-min; that is, the theorem holds in this case.

Suppose, on the other hand, that fewer than $I/36$ winner links minus $4$ per delete-min have a link rank at least $1$.  Then at least $I/12-I/36=I/18$ winner links have link rank at most $1$ and are of pass $1$, $2$, or $3$.

Let $vw$ be such a link.  We call $vw$ a \emph{type-1 link} if $v$ or $w$ has rank less than $1$ before the link, a \emph{type-2 link} otherwise.  By Lemma~\ref{L:small-rank-nodes}, the number of type-$1$ links is at most $(45/2)\lg n + 57/2$ per delete-min.  Thus the theorem holds if the number of type-$1$ links is at least $I/36$.

The remaining possibility is that at least $I/36$ of the winner links of pass 1, 2, or 3 with link rank less than $1$ are type $2$.  Let $vw$ be a type-$2$ link.  Then both $v$ and $w$ have rank at least $1$ just before the link.  By Corollary~\ref{C:small-link-rank}, doing link $vw$ reduces the rank of $v$ or $w$ or both by at least $1$.  By Theorem~\ref{T:node-rank-changes}, the sum of the node-rank decreases is $\OO(\lg\lg n)$ per decrease-key plus $\OO(\lg n)$ per delete-min plus $I/72$.  Since $1/36 > 1/72$, $I$ is $\OO(\lg\lg n \cdot \lg\lg\lg n)$ per decrease-key plus $\OO(\lg n)$ per delete-min; that is, the theorem holds in this case as well.
\end{proof}

\begin{theorem}\label{T:pure-pairing-heap-bounds}
In a sequence of pure pairing heap operations, each operation except delete-min takes $\OO(1)$ time worst-case; each delete-min takes $\OO(1)$ %\xxyedit{per operation} 
time plus $\OO(1)$ time per pairing link worst-case.  In such a sequence, the amortized number of pairing links done is $\OO(\lg\lg n\cdot\lg\lg\lg n)$ per decrease-key, $\OO(\lg n)$ per delete-min, and $\OO(1)$ per insertion.
\end{theorem}
\begin{proof}
The first part of the theorem is merely a restatement of the observation at the beginning of Section~\ref{S:terminology}.  The second part, our main result, is immediate from Theorem~\ref{T:outer-link-bound}, Lemma~\ref{L:group-pairing-passes}~(v), and \Cref{T:inner-link-bound}.  
\end{proof}

\section{Lazy pairing heaps}\label{S:lazy-pairing-heaps}

Our analysis applies almost without change to the lazy pairing heap described in Section~\ref{S:versions-of-pairing-heaps}.  The only significant change is in the definition of node groups. The list of roots representing a heap is partitioned into node groups, as are the children of each temporary node.  The leftmost group of roots in a heap and the leftmost group of children of a temporary node are growing; all other node groups are shrinking.  The only possibly empty groups are the growing groups.  When a heap is created, an empty growing group of its roots is created; when a temporary node is inserted, an empty leftmost growing group of its children is created.  When a meld occurs, the leftmost group of roots of the heap whose min-root is the min-root of the combined heap becomes the leftmost group of roots of the combined heap; the leftmost group of roots of the other heap becomes a shrinking group, or is deleted if it is empty.  When a delete-min occurs, the leftmost group of children of the deleted node becomes the leftmost group of roots; the previous leftmost group of roots, if any, becomes a shrinking group.  When a decrease-key of a node $v$ other than the min-root $u$ occurs, $u$ is deleted from its group, and $v$ is added to the leftmost group of roots, which may or may not be the one containing $u$ before the decrease-key.  The remaining details of the analysis are straightforward and mimic those of pure pairing heaps.  As mentioned at the end of Section~\ref{S:versions-of-pairing-heaps}, this analysis, appropriately modified, also applies to other versions of lazy pairing heaps, such as those in~\cite{FSST86, SODASmoothSlim}.

\section{Remarks}\label{S:remarks}

As one can easily verify, our analysis implies that the amortized and worst-case time to do decrease-key of a permanent node is $\OO(1)$: Every link won or lost by a permanent node is a phantom link, and such links do not change the rank of any node or link.  This is a general result that holds for multipass pairing heaps and slim and smooth heaps as well, by the analysis of~\cite{SinnamonT25}.

Our idea of partitioning nodes into groups has at least two additional applications. It applies to standard pairing heaps to give an improved amortized bound of $\OO((\log\log n)^2\log\log\log n)$ per decrease-key.  We shall present this result in a forthcoming companion paper.  It also is a key idea in proving that splay trees, a form of self-adjusting search tree, are almost dynamically optimal~\cite{Splay2026}.

There is still a gap of $\OO(\log\log\log n)$ between the upper and lower bounds for decrease-key in both multipass and pure pairing heaps, and at least this much in standard pairing heaps.  Reducing this gap remains a challenging open problem.

\paragraph{Declaration on AI Tool Use}
The authors used ChatGPT as a post-drafting review aid to flag possible calculation errors, notation inconsistencies, and gaps in reasoning. Any resulting revisions were written by the authors. The authors assume responsibility for all content.

\bibliographystyle{plainnat}
\bibliography{main}

\end{document}